\documentclass[psamsfonts]{amsart}

\usepackage{amsthm,amssymb,amscd,latexsym}

\def\R{{\mathbb{R}}}
\def\N{{\mathbb{N}}}
\def\Z{{\mathbb{Z}}}
\def \length {\ell}

\theoremstyle{plain}
\newtheorem{theorem}{Theorem}
\newtheorem{corollary}{Corollary}

\newtheorem{lemma}{Lemma}
\newtheorem*{fact*}{Fact}

\newtheorem{claim}{Claim}

\theoremstyle{definition}
\newtheorem{definition}{Definition}
\newtheorem{case}{Case}

\theoremstyle{remark}
\newtheorem{notation}{Notation}
\newtheorem{remark}{Remark}

\newcommand{\qedc}{{\qed}~{\rm Claim~{\theclaim}.}}

\numberwithin{equation}{section}
\DeclareMathOperator{\dist}{Dist}
\def \pn {\par \noindent}

\title{A rewrite based analysis of algorithms}

 \author[A. Akhavi, C. Moreira]{Ali Akhavi \and C\'eline
Moreira Dos Santos}
 \address{Computer Science dept.\\
 University of Caen\\
 F14032 Caen cedex\\
 FRANCE
 }
\email{ali.akhavi@info.unicaen.fr}
\urladdr{http://www.info.unicaen.fr/\~{}akhavi/}
\email{cmoreira@math.unicaen.fr}
\urladdr{http://www.math.unicaen.fr/\~{}cmoreira/}

\begin{document}

\begin{abstract}

We introduce here a new method for extracting worst--cases 
of algorithms
by using rewrite systems over automorphisms groups of 
inputs.
\pn We propose a canonical description of an algorithm, 
that is also
related to the problem it solves. The description 
identifies an 
algorithm with a set of a rewrite systems 
over the 
automorphisms groups of inputs. All possible execution of the 
algorithm will then be reduced words of these rewriting system.

\pn Our main example is reducing two-dimensional Euclidean 
lattice 
bases. We deal with the Gaussian algorithm that finds 
shortest vectors 
in a two--dimensional lattice. We introduce four rewrite 
systems in the 
group of unimodular matrices,
\emph{i.e.} matrices with integer entries and with
determinant equal to
$\pm 1$ and deduce a new worst-case analysis of the 
algorithm that
generalizes Vall\'ee's result\cite
{ValGaussRevisit}
to the case of the usual Gaussian algorithm.
An interesting (but not easy) future application will be 
lattice reduction
in higher dimensions, in order to exhibit a tight upper-
bound for the
number of iterations of LLL--like reduction algorithms in 
the worst case.

\pn Sorting ordered finite sets are here as a nice esay  example 
to 
illustrate the purpose of our method.
We propose several rewrite systems in the group $\mathcal 
S$ of 
permutations and canonically identify
a sorting algorithm with a rewrite system over $\mathcal 
S$. This brings us
to exhibit worst-cases of several sorting algorithms.

 \end{abstract}
\maketitle


\section{Introduction}

\pn A Euclidean lattice is the set of all integer linear
combinations
of a set of linearly independent vectors in ${\R}^p$. The
independent vectors
are called {\em a basis} of the lattice.
Any lattice can be generated by many bases. All of them
have the same
cardinality, that is called {\em the dimension} of the
lattice.
If $B$ and $B'$ represent matrices of two bases of the same
lattice in the
canonical basis of $\R^p$, then there is a unimodular
matrix $U$ such that
$B'=UB$. A unimodular matrix is a matrix with integer
entries and with
determinant equal to $\pm 1$.

\pn The lattice basis reduction problem is to find bases
with
good Euclidean properties, that is, with sufficiently short
vectors and
almost orthogonal.

\medskip
\pn In two dimensions, the problem is solved by the
Gaussian algorithm, that
finds in any two--dimensional lattice, a basis formed with
the shortest
possible
vectors. The worst--case complexity of Gauss' algorithm
(explained
originally in
the vocabulary of quadratic forms) was first studied by
Lagarias \cite{Lag1},
who showed that the algorithm is polynomial with respect to
its input. The
worst--case complexity of Gauss' algorithm was also studied
later more
precisely
by Vall\'ee\cite{ValGaussRevisit}.

\medskip
\pn In 1982, Lenstra, Lenstra and Lov\'asz \cite{LLL82}
gave a powerful
approximation reduction algorithm for lattices of
arbitrary dimension. Their
famous algorithm, called LLL, was an important breakthrough
to numerous
theoretical and practical problems in computational number
theory and
cryptography: factoring polynomials with rational
coefficients
\cite{LLL82}, finding linear Diophantine approximations
\cite{Lag0},
breaking various cryptosystems \cite{JoSt98} and integer
linear programming
\cite{Kann.G-Z,Len83}. The LLL algorithm is a possible
generalization of its
$2$--dimensional version, which is the Gaussian algorithm.

\pn The LLL algorithm seems difficult to analyze
precisely, both in the
worst--case\cite{Akh00,Len00,LLL82} and in average--
case\cite{DaFlVa,DaVa}.
In particular when the dimension is higher than two, the
problem of the real
worst--case of the algorithm is completely open. However,
LLL--like reduction
algorithms are so widely used in practice that the
analyzes are a real
challenge, both from a theoretical and practical point of
view.

\pn The purpose of our paper is a new approach to the worst-
-case analyze of
LLL--like lattice reduction algorithms. For the moment this
approach is
presented only in two dimensions. We have to observe here
that the worst
case of
some variant of the Gaussian algorithm is already known
\cite{ValGaussRevisit}.
Even if our paper generalize this knowledge to the case of
the usual Gaussian
algorithm, we do not consider it as the most important
point of this paper. Our
aim here is to present this new approach.

\pn An LLL--like lattice reduction algorithm uses some
(finite) elementary
transforms. We consider the group generated by these basic
transforms. Then we
exhibit a family of rewriting rules over this group,
corresponding to the
mechanism of the algorithm. The rewriting rules make some
forbidden sequences
and the length of a valid word over the set of generators
becomes very
close to
the number of steps of the algorithm. This makes appear
the smallest
length of
input demanding a given number of iterations to the
reduction algorithm.

\pn From a combinatorial point of view, the group of $n$--
dimensional lattice
transformations $GL_n(\Z)$, \emph{i.e.} the multiplicative
group of
$n\times n$
matrices with determinant $\pm 1$, is the group of
automorphisms of the free
Abelian group on $n$ free generators\footnote{By the free
Abelian group, we
mean
that the only non trivial relators (\emph{i.e.} the
additional relators
compared
to the free group) are the commutators.}. Here we are
concerned by $GL_2(\Z)$,
which is is well--known and whose presentation in terms
of generators and
relators is known since the nineteenth century.

\pn In this paper we present a rewriting system over $GL_2
(\Z)$, that makes us
{\em predict} how the Gaussian algorithm is running on an
arbitrary input. We
deduce from this the worst--case configuration of the usual
Gaussian algorithm
and give an ``optimal'' maximum for the number of steps of
the Gaussian
algorithm. Our result generalizes the result of Vall\'ee
\cite{ValGaussRevisit}.
She studied a variant of the Gaussian algorithm where
elementary transforms
made
by the algorithm are some integer matrices of determinant
equal to $1$. In the
case of the usual Gaussian algorithm, elementary transforms
are integer
matrices
of determinant either $1$ or $-1$.
\pn Il the following we briefly outline the two steps of our method.
\subsection{First step}
\pn Consider a deterministic algorithm $A$ that run on an input $x$ (the data $x$ is  a set $ X$ of data). Then by mean  of elementary transforms taken in a set $F \subset X^X$, the algorithm  changes the input step by step ($x \to f(x)$  until the modified  data satisfies
some output condition ($x \in  O \subset  X$. An elementary or atomic  transform is a transform that cannot be decomposed by the algorithm:
\begin{equation}\label{atomic}
\forall f \in F, \quad \forall k \in \N, k>1, \quad\forall (f_1,\dots,f_k) \in (F\backslash \{id\})^k , \quad f\neq \prod_{i=1}^k f_i
\end{equation}
This is of course a very general context containing both iterative and recursive algorithms.
\pn {\bf Algorithm $A$ }:

\par {\bf Input:} $x \in X$.
\par{\bf Output:} $y \in O$.

\pn {\bf Initialization:} $\mathbf{i:=1;}$ \ 
\pn {\bf While $x \notin O$ do}
\par Determine an adequate function $f \in F$ by a computation on $x$ and eventually on $i$.
\par x:=f(x)
\par i:=i+1

\pn The algorithm $A$ is deterministic.  We suppose that the determination of  the adequate function $f \in F$ at a moment $i$ (which may depend on the history of the execution)   has a cost that  can be added without ambiguity  to the cost of the function $f$ . Considering $F$ as an alphabet, the set $F^*$ of finite words on $F$ is then the monoid generated by the set of {\em free} generators $F$.  $F^*$ contain all (finite) executions of the algorithm.

Now fix  a sequence of transforms $(f_1,f_2,\dots, f_k) \in  F^k$. 

\begin{quotation}\label{q:general}{\em Is the sequence  $(f_1,f_2,\dots, f_k) \in  F^k$ a possible execution for the algorithm? } More precisely, is there $(x,y) \in X \times Y$ such that the algorithm $A$ outputs $y$ when running on an input $x$ and following the exact sequence of transforms $(f_1,f_2,\dots, f_k)$? 
\end{quotation}
\pn Answering  this question in such a  general context is very difficult and the general problem (formulated more precisely) is likely undecidable. However the answer in restricted class of algorithms bring indeed a strong understanding of the mechanism of the algorithm and we believe that it is interesting by its own with lots of possible applications in program verifying, or program designing.

\pn In this paper, we propose a method to answer this question in the case of Gaussian algorthm and three sorting algorithms. A set (finite in the case of our examples) of rewriting systems encode all possible executions of a given algorithm.  All possible executions will be the normal forms (or reduced forms) of these rewriting systems. 
\subsection{second step}

\pn Usually when counting the number of steps of an
algorithm, one
considers all
inputs of length less than a fixed bound, say $M$. Then one
estimates the
maximum number of steps taken over all these inputs by:
 \begin{equation}\label{majorant}
 f(M):= \max_{\text{all inputs of length less than $M$}}
 \text{number of steps of the algorithm}.
 \end{equation}\footnote{When dealing with a non--trivial
algorithm $f$ is
always an increasing function.}

\pn Here to exhibit the precise real worst--case, we first
proceed in ``the
opposite way''. Consider $k$ a fixed number of steps. We
will estimate the
minimum length of those inputs demanding at least $k$ steps
to be processed by
the algorithm:
 \begin{equation}\label{recip}
 g(k):= \min_{\text{all inputs demanding at least $k$
steps}}
 \text{length of the input}.
 \end{equation}
Clearly $f(g(k))=k$. Otherwise there would be an input of
length less than
$g(k)$ demanding more than $k$ steps. But $g(k)$ is by
definition the minimal
length of such inputs. So by inverting the fuction $g$ , we can compute
$f$.

\section{Gaussian algorithm and the new approach to its
worst--case analysis}

\pn Endow $\R^2$ with the usual scalar product
$( \: ,\: )$ and Euclidean length
$|\mathbf{u}| = {( \mathbf{u} , \mathbf{u})}^{1/2}$.
A two--dimensional lattice is a discrete additive subgroup
of $\R^2$.
Equivalently, it is the set of all integer linear
combinations of two
linearly independent vectors. Generally it is given by one
of its bases
$(\mathbf{b}_{1}, \mathbf{b}_{2})$.
Let $(\mathbf{e}_{1}, \mathbf{e}_{2})$ be the canonical 
basis
of $\R^2$. We often
associate to a lattice basis $(\mathbf{b}_{1}, \mathbf{b}_
{2})
$ a matrix $B$,
such that {\em the vectors of the basis are the rows of the
matrix}:
\begin{equation}\label{rows}
B= \bordermatrix{
 &\mathbf{e}_1 &\mathbf{e}_2 \cr
\mathbf{b}_1 &b_{1,1} &b_{1,2} \cr
\mathbf{b}_{2} &b_{2,1} &b_{2,2} \cr
}.
\end{equation}
The \emph{length $L$ of the previous basis} (or the \emph
{length of the matrix
$B$}) is defined here to be the maximum of $(|\mathbf{b}
_1|,|\mathbf{b}_2|)$.

\noindent The usual Gram--Schmidt orthogonalization process
builds, in
polynomial--time, from a basis $b=(\mathbf{b}_{1}, \mathbf
{b}_
{2})$ an orthogonal
basis
$b^{*}=(\mathbf{b}^{*}_{1}, \mathbf{b}^{*}_{2})$
and a lower--triangular matrix $M$ that expresses the
system $b$ into the
system
$b^{*}$\footnote{ Of course,
$b^*$ is generally not a basis for the lattice generated by
$b$.}.
Put
$m=\frac{(\mathbf{b}_2,\mathbf{b}_1)}{(\mathbf{b}_1,\mathbf
{b}_1)}$.
By construction, the following equalities hold:
 \begin{equation}\label{b*}
 \left\{ \begin{array}{lll}
 \mathbf{b}^*_1 &= &\mathbf{b}_1\\
 \mathbf{b}^*_2 &= &{\mathbf{b}_2} - m\:\mathbf{b}_1
 \end{array} \right. ,\
 M= \bordermatrix{
 &\mathbf{b}^*_1 &\mathbf{b}^*_2 \cr
 \mathbf{b}_1 &1 &0 \cr
 \mathbf{b}_{2} &m &1 \cr
 }.
 \end{equation}

\pn The ordered basis $B=(\mathbf{b}_1,\mathbf{b}_2)$ is 
called
{\em proper }
if the
quantity $m$ satisfies
 \begin{equation}\label{proper}
 -1/2 \leq m < {1}/{2}.
 \end{equation}

\pn There is a natural and unique representative of all the
bases of given
two--dimensional lattice. This basis is composed of two
shortest vectors
generating the whole lattice. It is called the Gauss--
reduced basis and the
Gaussian algorithm outputs this reduced basis running on
any basis of the
lattice. Any lattice basis in two dimensions can always be
expressed as
 \begin{equation}\label{deb}
 B=U\: R,
 \end{equation}
where $R$ is the so--called Gaussian reduced basis of the
same lattice and $U$
is a unimodular matrix, \emph{i.e.} an element of $GL_2(\Z)
$. The goal of a
reduction algorithm, the Gaussian algorithm in two
dimensions, is to find $R$
given $B$. The Gaussian algorithm is using two kinds of
elementary transforms,
explained in the sequel of this paper. Let $(b_1, b_2)$ be
an input basis of a
lattice and the matrix $B$ expressing $(b_1, b_2)$ in the
canonical basis of
$\R^2$ as specified by \eqref{rows}.

\pn The algorithm first makes an integer translation of $b_2
$ in the direction
of $b_1$ in order to make $b_2$ as short as possible. This
is done just by
computing the integer $x$ nearest to $m=(b_2,b_1)/(b_1,b_1)
$ and replacing
$b_2$
by $b_2 - x b_1$. Notice that, after this integer
translation, the
basis $(b_1,b_2)$ is proper.

\pn The second elementary transform is just the swap of the
vectors $b_1$ and
$b_2$ in case when after the integer translation we have
$|b_1|>|b_2|$.

\pn The algorithm iterates these transforms, until
after the
translation, $b_1$ remains still smaller than $b_2$, \emph
{i.e.},
$|b_1|\leq|b_2|$.

\pn The Gaussian algorithm can also be regarded (especially
for the analysis
purposes) as an algorithm that gives a decomposition of the
unimodular matrix
$U$ of relation \eqref{deb} by means of some basic
transforms:
 \begin{equation}\label{mieux}
 \begin{array}{lll}
 \text{Input:}&& B=U\: R.\\
 \text{Output:}&&R=T^{x_{k+1}}ST^{x_k}ST^{x_{k-1}}\dots
 ST^{x_2}ST^{x_{1}}B;
 \end{array}
 \end{equation}
where
 \begin{equation}\label{ST}
 S=\begin{pmatrix}
 0&1\\
 1&0
 \end{pmatrix}
 \quad \text{and}\quad
 T=\begin{pmatrix}
 1&0\\
 1&1
 \end{pmatrix}.
 \end{equation}
The matrix $T$ corresponds to an integer translation of $b_2
$
in the direction of $b_1$ by one.
Of course, we have:
 $$
 T^x=\prod_{i=1}^x \text{ and }T=\begin{pmatrix}1&0\\ x&1
\end{pmatrix},
 \text{ for all }x\in\Z.
 $$
The matrix $S$ represents a swap.
{\em Each step of the algorithm is indeed an integer
translation followed by a
swap}\footnote{A priori $x_1$ and $x_{k+1}$ in \eqref
{mieux} may be zero so the
algorithm may start by a swap or finish by a translation.}.
So each step of the Gaussian algorithm is represented by
$ST^x$,
$x\in \Z^*$.

\pn Writing the output in this way \eqref{mieux} shows not
only the output but
how precisely the algorithm is working since $T$ and $S$
represent the only
elementary transforms made during the execution of the
Gaussian algorithm.

\pn So when studying the mechanism of a reduction algorithm
in two
dimensions and for a fixed reduced basis $R$,  the algorithm can be regarded as a decomposition
algorithm over
$GL_2(\Z)$:
 \begin{equation}\label{enfin}
 \begin{array}{lll}
 \text{Input:}&&U \in GL_2(\Z).\\
 \text{Output:}&&\text{a decomposition of $U$, }
 U:=T^{x_{k+1}}\,ST^{x_k}\,ST^{x_{k-1}}\dots ST^{x_2}
\,ST^{x_{1}}.
 \end{array}
 \end{equation}
The integer $k$ denotes the number of steps. Indeed the
algorithm terminates\cite{Akh00,Lag1,ValGaussRevisit}. In
the sequel we will prove  that the above mechanism does not depend strongly
on the reduced basis $R$. More precisely there are exactly $4$ rewrite systems.
(for all reduced bases of $\R^2$) 
\medskip
\pn The unimodular group in two dimensions $GL_2(\Z)$ has
been already studied
\cite{Mag34,MKS76,Nie24,Nie24b} and it is well--known that 
$\{S,T\}$ is a
possible family of
generators for $GL_2(\Z)$. Of course there are relators
associated to these
generators and there is no unicity of the decomposition of
an element of
$GL_2(\Z)$ in terms of
$S$ and $T$. But the Gaussian algorithm gives one precise
of these possible
decompositions. In the sequel of this paper, we will
completely characterize
this decomposition and we will call it {\em the Gaussian
decomposition of a
unimodular matrix}. Roughly speaking, we exhibit forbidden
sequences of values
for the $x_i$--s.

\pn More precisely, we exhibit in~Section~\ref{S:combgroup}
a set of rewriting
rules that lead to the formulation output by the Gaussian
algorithm, from
any product of matrices involving $S$ and $T$. The
precise characterization of the Gaussian decomposition that
we give makes
appear
the slowest manner the length of a unimodular matrix can
grow with respect to
its Gaussian decomposition. More precisely we consider
unimodular matrices
the length of Gaussian decomposition of which is fixed, say
$k$:
 $$
 U:=T^{x_{k+1}}\, ST^{x_k}\,ST^{x_{k-1}}\dots ST^{x_2}
\,ST^{x_{1}}.
 $$
We exhibit in Section~\ref{S:increase} the Gaussian word of 
length $k$ with minimal length.
We naturally deduce the minimum length $g(k)$
of all inputs demanding at least $k$ steps (Section \ref
{S:result}). Finally by
``inverting'' the function $g$ we find the maximum number
of steps of the
Gaussian algorithm.


\section{The Gaussian decomposition of a unimodular matrix}
\label{S:combgroup}
Let $\Sigma$ be a (finite or infinite) set. A word $\omega$
on $\Sigma$ is
a finite sequence
 \begin{equation}\label{longueur}
 \alpha_1\alpha_2\dots\alpha_n
 \end{equation}
where $n$ is a positive integer, and $\alpha_i \in \Sigma$,
for all
$i\in\{1,\dots,n\}$. Let $\Sigma^*$ be the set of finite
words on
$\Sigma$. We introduce for convenience the {\em empty
word} and we
denote it by $1$.
\pn Consider the alphabet $\Sigma = \{S,T,T^{-1}\}$.
Remember that we call Gaussian decomposition, the
expression of $U$ output by the Gaussian algorithm.
Remember also that, for
a given
basis $B$, there exists a unique couple $(U,R)$ such that
$U$ and $R$ are
output by
the Gaussian algorithm while reducing $B$.

\begin{lemma}
Let $R=(b_1,b_2)$ be a reduced basis. Then one of the
following cases occurs:
\begin{itemize}
\item $|b_1|<|b_2|$ and $m\ne-1/2$;\label{E:dsdom}
\item $|b_1|=|b_2|$ and $m\ne-1/2$;\label{E:bddomg}
\item $|b_1|<|b_2|$ and $m=-1/2$;\label{E:bddomc}
\item $|b_1|=|b_2|$ and $m=-1/2$.\label{E:bddompt}
\end{itemize}
\end{lemma}

Now consider a word $\omega$ on $\Sigma$, that is, an
element of
$\Sigma^*$, a unimodular matrix $U$ and a reduced basis
$R$. We give, in the
following subsections, sets of rewriting rules depending on
the form of
$R$, such that
any word in which none of these rewriting rules can be
applied is
Gaussian. Since
the results of these subsections are very similar, we only
give detailed proofs for Subsection~\ref{Ss:dsd} in the 
appendix.

\subsection{The basis $R$ is such that $|b_1|<|b_2|$ and
$m\ne-1/2$}\label{Ss:dsd}
Say that $\omega$ is a {\em reduced word} or a {\em
reduced
decomposition of the
unimodular matrix $U$}, if $\omega$ is
a decomposition
of $U$ in
which no one of the rewriting rules of Theorems~\ref{rules}
can be applied.

\pn Thus, Theorem~\ref{rules} shows that, for a given
reduced basis $R$
such that $|b_1|<|b_2|$ and $m\ne-1/2$, the Gaussian
decomposition and a
reduced
decomposition of a unimodular matrix are the same, which
implies that this
decomposition is unique.

\begin{theorem}\label{rules}
Let $\omega_1$ be any decomposition of $U$ in terms of the
family of generators $\{S,T\}$. The Gaussian
decomposition of $U$ is
obtained
from $\omega_1$ by applying repeatedly the following set of
rules:
 \begin{equation}\label{1}
 S^2 \longrightarrow 1;
 \end{equation}
 \begin{equation}\label{2}
 T^{x}T^{y}\longrightarrow T^{x+y};
 \end{equation}
 \begin{equation}\label{3}
 \forall x \in \Z_-^*, \qquad\qquad
 ST^{2}ST^{x}\mbox{\ }\longrightarrow
 TST^{-2}ST^{x+1};
 \end{equation}
 \begin{equation}\label{4}
 \forall x \in \Z_+^*,\qquad\qquad ST^{-2}ST^{x}
\longrightarrow
 T^{-1}ST^{2}ST^{x-1};
 \end{equation}
 \begin{equation}\label{5}
 \forall x \in \Z^*, \forall k\in \Z_+,\
 ST^{}ST^{x}\prod_{i=k}^1 ST^{y_{i}}
 \longrightarrow
 T^{}ST^{-x-1}\prod_{i=k}^1 ST^{-y_{i}};
 \end{equation}
 \begin{equation}\label{6}
 \forall x \in \Z^*, \forall k\in \Z_+,\
 ST^{-1}ST^{x}\prod_{i=k}^1 ST^{y_{i}}\longrightarrow
 T^{-1}ST^{-x+1}\prod_{i=k}^1 ST^{-y_{i}}.
 \end{equation}
\end{theorem}

\pn The trivial rules \eqref{1} and \eqref{2} have to be
applied whenever
possible. So any word $\omega_1$ on the alphabet $\Sigma$
can trivially be written as
 \begin{equation}\label{bondepart}
 T^{x_{k+1}}\prod_{i=k}^1 ST^{x_{i}},
 \end{equation}
with $x_i\in \Z^*$ for $2\leq i\leq k$ and $(x_1,x_{k+1})
\in \Z^2$. The
integer
$k$ is called
{\em the length\footnote{The length of a word has of course
to be
distinguished
with what we call the length of a unimodular matrix, that
is the maximum of
absolute values of its coefficients.} of $\omega_1$.}
Notice that usually the length of a word as in
\eqref{longueur} is $n$, which would corresponds here to
$2k+1$. Here the
length
is $k$, which correspond to the number of iterations of
the algorithm minus 1.

The proof of Theorem~\ref{rules} is given in appendix. It
consists in the following lemmas:
\begin{itemize}
\item Lemma~\ref{L:rptGL}, where we prove that the
rewriting process terminates;
\item Lemma~\ref{L:redgaus}, where we prove that any
reduced word is also Gaussian;
\item Lemmas~\ref{L:redST2STx} and~\ref{L:redSTSTx}, where
we prove that the use of a nontrivial rewriting rules
changes a base of the lattice in another base of the same
lattice.
\end{itemize}
Proofs of Theorems~\ref{rules2},~\ref{rules3} and~\ref
{rules4}, which are given in the following subsections, are
very similar.

\subsection{The basis $R$ is such that $|b_1|=|b_2|$ and
$m\ne-1/2$}

\begin{theorem}\label{rules2}
Let $\omega_1$ be any decomposition of $U$ in terms of the
family of generators $\{S,T\}$. The Gaussian
decomposition of $U$ is
obtained
from $\omega_1$ by applying repeatedly the set of rules
\eqref{1} to
\eqref{6} of Theorem~\ref{rules}, together with the 
following
rules:
 \begin{equation}\label{53}
 \forall x \in \Z^*, \forall k\in \Z_+,\
 ST^{}ST^{x}\left(\prod_{i=k}^1 ST^{y_{i}}\right)
 \longrightarrow
 T^{}ST^{-x-1}\left(\prod_{i=k}^1 ST^{-y_{i}}\right)T;
 \end{equation}
 \begin{equation}\label{63}
 \forall x \in \Z^*, \forall k\in \Z_+,\
 ST^{-1}ST^{x}\left(\prod_{i=k}^1 ST^{y_{i}}\right)
\longrightarrow
 T^{-1}ST^{-x+1}\left(\prod_{i=k}^1 ST^{-y_{i}}\right)T;
 \end{equation}
 \begin{equation}\label{83}
 \text{ if } \omega_1=\omega\,S, \text{ then }
 \omega S \longrightarrow \omega ;
 \end{equation}
 \begin{equation}\label{93}
 \text{ if } \omega_1=\omega\,ST, \text{ then }
 \omega ST\longrightarrow \omega TST^{-1}.
 \end{equation}
\end{theorem}

\subsection{The basis $R$ is such that $|b_1|<|b_2|$ and
$m=-1/2$}

\begin{theorem}\label{rules3}
Let $R$ be a reduced basis and let $U$ be a unimodular
matrix, \emph{i.e.}, an
element of $GL_2(\Z)$. Let $\omega_1$ be any decomposition
of $U$ in terms
of the
family of generators $\{S,T\}$. The Gaussian
decomposition of $U$ is
obtained
from $\omega_1$ by applying repeatedly the rules \eqref{1}
to \eqref{4} of
Theorem~\ref{rules} until $\omega_1$ is reduced in the
sense of Theorem~\ref{rules}. Then, if we have
$\omega_1=\omega\,ST^2\,S$,
the following rule applies:
 \begin{equation}\label{7}
 \omega ST^2\,S \longrightarrow \omega T\,ST^{-2}\,ST,
 \end{equation}
and the rewriting process is over.
\end{theorem}

\subsection{The basis $R$ is such that $|b_1|=|b_2|$ and
$m=-1/2$}

\begin{theorem}\label{rules4}
Let $R$ be a reduced basis and let $U$ be a unimodular
matrix, \emph{i.e.}, an
element of $GL_2(\Z)$. Let $\omega_1$ be any decomposition
of $U$ in terms
of the
family of generators $\{S,T\}$. The Gaussian
decomposition of $U$ is
obtained
from $\omega_1$ by applying repeatedly Rules~\eqref{1}
to~\eqref{4} of
Theorem~\ref{rules}, together with Rules~\eqref{53} 
and~\eqref{63} and the
following set of
rules:

 \begin{equation}\label{101}
 \text{ if }\omega_1=\omega\,S,\text{ then }
 \omega S \longrightarrow \omega ;
 \end{equation}
 \begin{equation}\label{102}
 \text{ if }\omega_1=\omega\,ST,\text{ then }
 \omega ST\longrightarrow \omega T;
 \end{equation}
 \begin{equation}\label{103}
 \text{ if }\omega_1=\omega\,ST^2,\text{ then }
 \omega ST^2\longrightarrow \omega TST^{-1}.
 \end{equation}
\end{theorem}

\section{The length of a unimodular matrix with respect to
\\its Gaussian decomposition}\label{S:increase}

\pn Let $B=(b_1,b_2)$ be a basis. The \emph{length} of $B$,
denoted by
$\length(B)$, is
the sum of the squares of the norms of its vectors, that is,
$\length(B)=|b_1|^2+|b_2|^2$.

The easy but tedious proof of the following theorem is
given in the appendix, see Lemmas~\ref{L:sizemin1},~\ref
{L:sizemin2},~\ref{L:sizemin3},~\ref{L:sizemin4} and~\ref
{L:sizemin5}.

\begin{theorem}\label{T:basemin}
Let $R=(b_1,b_2)$ be a reduced basis, let $k$ be a positive
integer, and let
$x_1$, \dots, $x_{k+1}$ be integers such that the word
$\omega=T^{x_{k+1}}\prod_{i=k}^1 ST^{x_{i}}$ is Gaussian.
Then
the following properties hold:
\begin{enumerate}
\item if $|b_1|<|b_2|$ and $m\geq 0$ then
$\length(\omega R)\geq \length((ST^{-2})^{k-1}\,S\,R)$;
\item if $|b_1|<|b_2|$ and $-1/2<m<0$ then
$\length(\omega R)\geq \length((ST^{2})^{k-1}\,S\,R)$;
\item if $|b_1|<|b_2|$ and $m=-1/2$ then
$\length(\omega R)\geq\length((ST^{-2})^{k-1}\,ST\,R)$;
\item if $|b_1|=|b_2|$ then
$\length(\omega R)\geq\length((ST^{-2})^{k-1}\,ST^{-1}\,R)$.
\end{enumerate}
\end{theorem}

%
%


\section{The maximum number of steps of the Gaussian
algorithm}\label{S:result}

\begin{theorem}\label{T:compirecas}
\pn Let $k>2$ be a fixed integer. There exists an absolute
constant $A$ such
that input basis demanding more than $k$ steps to the
Gaussian algorithm has a
length greater than $A (1+\sqrt{2})^k$:
$$g(k)\geq A (1+\sqrt{2})^k.$$
\end{theorem}

\pn It follows that any input with length less than $A
(1+\sqrt{2})^k$ is
demanding less than $k$ steps. We deduce the following
corollary.

\begin{corollary}\label{C:conspircas}
There is an absolute constant $A$ such that the number of
steps of the
Gaussian algorithm on inputs of length less than $M$ is
bounded from above by
\[ \log_{(1+\sqrt 2)} \left(\frac{M}{A}\right).\]
\end{corollary}



\section{Sorting algorithms}
Let $n$ be a positive integer, and let $[1,\dots,n]$ be the
sorted list of the $n$ first positive integers. Let
$\mathcal{S}_n$ be the set of all permutations on
$[1,\dots,n]$, and let $\mathcal{S}$ be the set of all
permutations on a list of distinct integers of variable
size. Let us denote by $t_i$ the transposition which swaps
the elements in positions $i$ and $i+1$ in the list , for
all $i\in\{1,\dots,n\}$. Any permutation can be written in
terms of the $t_i$-s. Put $\Sigma_n=\{t_1,\dots,t_n\}$ and
$\Sigma=\{t_i\colon i\in\N^*\}$. Thus $\Sigma_n$ (resp.
$\Sigma$) is a generating set of $\mathcal{S}_n$ (resp.
$\mathcal{S}$).

As in previous sections, any word $\omega$ on $\Sigma$ will
be
denoted as following:
 \[
 \omega= t_{i_1} t_{i_2} \dots t_{i_k}=\prod_{j=1}^{k} t_
{i_j},
 \]
where $k$ and $i_1$, \dots, $i_k$ are positive integers.

\begin{definition}
Let $\omega_1=t_{i_1}t_{i_2}\dots t_{i_k}$ and $\omega_2=t_
{j_1}t_{j_2}\dots t_{j_l}$ be words on $\Sigma$.
 \begin{enumerate}
 \item The \emph{length of $\omega$}, denoted by
$|\omega|$, is $k$;
 \item the \emph{distance between $\omega_1$ and
$\omega_2$}, denoted by
 $\dist(\omega_1,\omega_2)$, is given
 by $\min_{t_i\in\omega_1,t_j\in\omega_2}|i-j|$;
 \item the \emph{maximum} (resp. \emph{minimum}) of
$\omega_1$, denoted
 by $\max(\omega_1)$ (resp. $\min(\omega_1)$), is given
 by $\max_{t_i\in\omega_1}(i)$ (resp. $\min_
{t_i\in\omega_1}(i)$);
 \item $\omega_1$ is an \emph{increasing word} (resp.
\emph{decreasing word})
 whether $i_p<i_{p+1}$ (resp. $i_p>i_{p+1}$), for all
$p\in\{1,\dots,k-1\}$;
 \item $\omega_1$ is a \emph{consecutively increasing
word}
 (resp. \emph{consecutively decreasing word}) whether $i_
{j+1}-i_j=1$
 (resp. $i_j-i_{j+1}=1$), for all $j\in\{1,\dots,k-1\}$;
 \item $\omega_1<\omega_2$ (resp. $\omega_1>\omega_2$)
whether
 $\omega_1$ and $\omega_2$ are increasing (resp.
decreasing) words such that
 $\max(\omega_1)\leq\min(\omega_2)$ (resp. $\min
(\omega_1)\geq\max(\omega_2))$;
 \item $\omega_1$ is \emph{minimal on the left in
$\omega_2\omega_1$} (resp.
 \emph{maximal on the right in $\omega_1\omega_2$})
whether $\omega_1$ is an
 increasing word such that $i_1\leq j_l$ (resp. $j_1\leq
i_k$);
 \item similarly, $\omega_1$ is \emph{minimal on the
right in $\omega_2\omega_1$} (resp.
 \emph{maximal on the left in $\omega_1\omega_2$})
whether $\omega_1$ is a
 decreasing word such that $j_l\leq i_1$ (resp. $i_k\leq
j_1$).
 \end{enumerate}
\end{definition}

It is easy to prove that any word $\omega$ on $\Sigma$ can
be uniquely written on the form
 \begin{equation}\label{permdecomp}
 \omega=\omega_1\omega_2\dots\omega_m,
 \end{equation}
where $\omega_i$ is an increasing (resp. decreasing) word
maximal on the right and on the left, for all
$i\in\{1,\dots,m\}$. We will call \eqref{permdecomp}
the \emph{increasing
decomposition} (resp. \emph{decreasing decomposition}) of
$\omega$, and we will denote it by
$[\omega_1,\dots,\omega_m]$.
We define $s\colon\Sigma^*\to\N$ as the map given by the
rule
 \[
 s(\omega)=m,
 \]
where $[\omega_1,\dots,\omega_m]$ is the increasing
decomposition of $\omega$.
Moreover, it is also easy to prove that $\omega_i$ can be
uniquely written on the form
 \[
 \omega_i=\omega'_{i_1}\omega'_{i_2}\dots\omega'_{i_{p_i}},
 \]
where the $\omega'_j$-s are consecutively increasing (resp.
decreasing) and minimal on the left (resp. right), for all
$i\in\{1,\dots,m\}$. The decomposition $[\omega'_
{i_1},\omega'_{i_2},\dots,\omega'_{i_{p_i}}]$
is called the \emph{consecutively increasing decomposition}
(resp. \emph{consecutively decreasing decomposition}) of
$\omega_i$.

\subsection{Bubble sort}\label{Ss:bubble}
The basic idea of the bubble sort algorithm is the
following: pairs of adjacent values in the list to be
sorted are compared and interchanged if they are out of
order, the process starting from the beginning of the list.
Thus, list entries `bubble upward' in the list until they
bump into one with a higher sort value.

The algorithm first compares the two first elements of the
list and swap them if they are in the wrong order. Then,
the algorithm compares the second and the third elements of
the list and swaps them if necessary. The algorithms
continues to compare adjacent elements from the beginning
to the end of the list. This whole process is iterated
until no changes are done.

Let $\sigma$ be a permutation on $[1,\dots,n]$. There
exists a unique decomposition $\omega$ of $\sigma$ on the
alphabet $\Sigma$ corresponding to the sequence of
elementary transforms performed by the bubble sort
algorithm on $\sigma[1,\dots,n]$. We will call it the \emph
{bubblian decomposition} of $\sigma$. Notice that $(\omega)^
{-1}\sigma=1$.

 \begin{equation}\label{bubble}
 \begin{array}{lll}
 \text{Input:}&&\sigma \in \mathcal{S}.\\
 \text{Output:}&&\text{a decomposition of $\sigma$, }
 \sigma:=t_1\dots t_m.
 \end{array}
 \end{equation}

\begin{definition}
A word $\omega$ on $\Sigma$ is a \emph{bubblian word} if it
corresponds to a possible execution of the bubble sort
algorithm.
\end{definition}

Let us define some rewriting rules on $\Sigma^*$. In the
following equations, $i$, $j$ and $k$ are arbitrary
positive integers and $\omega$ is a word on $\Sigma$:
 \begin{equation}\label{perm1}
 t_i\,t_i \longrightarrow 1;
 \end{equation}
 \begin{equation}\label{perm2}
 \text{ if } \dist(i+1,\omega)>1, \text{ then }
 t_{i+1}\,\omega\,t_{i}\,t_{i+1}\longrightarrow \omega
t_{i}\,t_{i+1}\,t_{i};
 \end{equation}
 \begin{equation}\label{perm3}
 \text{ if } \dist(i+1,\omega)>1 \text{ and }\omega \text
{ is maximally increasing, then }
 \omega\,t_{i}\longrightarrow t_i\,\omega;
 \end{equation}
 \begin{equation}\label{perm4}
 \text{ if } \dist(j,k\omega)>1 \text{ and either }
 i\leq j\leq k
 \text{ or }k<i\leq j,
 \text{ then }
 t_i\,t_k\,\omega\,t_{j}\longrightarrow
t_i\,t_j\,t_k\,\omega.
 \end{equation}

\begin{theorem}\label{T:redbubble}
Let $\sigma$ be a permutation and let $\omega\in\Sigma^*$
be a decomposition of $\sigma$ on $\Sigma$. The bubblian
decomposition of $\sigma$ is obtained from $\omega$ by
applying
repeatedly the rules \eqref{perm1} to \eqref{perm4}.
\end{theorem}

\begin{remark}
Let $\omega$ and $\omega'$ be words on $\Sigma$.
It is well known that a presentation of $\mathcal{S}$ on
$\Sigma$ is the following:
\begin{itemize}
\item $t_i\,t_i=1$;
\item $t_i\,t_j=t_j\,t_i$;
\item $t_i\,t_{i+1}\,t_i=t_{i+1}\,t_i\,\,t_{i+1}$;
\end{itemize}
for all positive integers $i$, $j$ such that $|i-j|=1$.
Thus,
it is easy to prove that if $\omega'$ is obtained from
$\omega$, then $\omega=\omega'$ in $\mathcal{S}$.
\end{remark}

The sketch of the proof of Theorem~\ref{T:redbubble} is
very similar to the proof of Theorem~\ref{rules} and is
given in the appendix. Notice that we can easily deduce 
from Theorem~{T:redbubble} the worst-case for the bubble 
sort algorithm.

\subsection{Other iterative sorting algorithms}
We also give without proof some rewriting rules for the
insertion sort algorithm and the selection sort algorithm,
see the appendix.

\section{Conclusion}

\pn In this paper we studied the Gaussian algorithm by
considering a rewriting
system over $GL_2(\Z)$. We first believe that our method
should be applied to
other variants of the Gaussian algorithm (for example,
Gaussian algorithm with
other norms \cite{KaSc96}) and for each variant there is an
adequate rewriting
system over $GL_2(\Z)$. 
\pn The most important and interesting continuation to this
work is to
generalize the approach in higher dimensions. Even in three
dimensions the
worst--case configuration of all possible generalization
of the Gaussian
algorithm is completely unknown for the moment. (\cite
{vSpr} has tried without
success in $3$ dimensions.) Although the problem is really
difficult, we have
already achieved a step, since the LLL algorithm uses the
Gaussian
algorithm as
an elementary transform.

\pn The group of $n$--dimensional lattice transformations
has been studied
first by Nielsen \cite{Nie24} ($n=3$) and for an arbitrary
$n$ by Magnus
\cite{Mag34,MKS76}, based on the work of Nielsen\cite
{Nie24b}. Their work
should certainly help to exhibit such rewriting systems on
$GL_n(\Z)$ if there
exists.

\pn This approach may also be an insight to the still open
problem of the
complexity of the optimal LLL algorithm\cite{Akh00,Len00}.

\pn {\bf Acknowledgments. }
The authors are indebted to Brigitte Vall\'ee for drawing
their
attention to algorithmic
problems in lattice theory and for regular helpful
discussions.
\medskip

\renewcommand{\baselinestretch}{0.9}

\appendix
\section{The Gaussian decomposition of a unimodular matrix}
\pn In the following proofs, we will essentially use the
set of rules
\eqref{3}, \eqref{4}, \eqref{5}, \eqref{6} and apply \eqref
{1} and \eqref{2}
implicitly whenever possible.

\pn So with any initial $\omega_1$, we always obtain a
reduced word after
applying a finite number of times the rewriting rules.
Moreover, no matter in
which order the different rewriting rules are used, the
same unique reduced
word -- corresponding to a Gaussian word -- is always
obtained from
$\omega_1$, as proved by the following lemmas.

\begin{lemma}\label{L:rptGL}
Let $\omega_1$ be a word as in \eqref{bondepart}. Then the
rewriting process always terminates\footnote{Of course
saying that the
rewriting
process presented
by the
previous Theorem always terminates has a priori nothing to
do with the
well-known fact that the Gaussian algorithm always
terminates.}.
\end{lemma}

The proof of Lemma~\ref{L:rptGL} will use the following
notations.

\begin{notation}
Let k be a nonnegative integer, and let $x_1$, \dots, $x_
{k+1}$ be integers
such
that $x_2$,\dots, $x_k$ are nonzero. Put
$\omega_1=T^{x_{k+1}}\prod_{i=1}^k ST^{x_{i}}$. We denote by
$\omega_1^-$ the word $T^{-x_{k+1}}\prod_{i=1}^k ST^{-x_{i}}
$. Put
\begin{align*}
&S_1=\{i\colon 2\leq i\leq k \text{ and } |x_i|=1\};\\
&S_2=\{i\colon 2\leq i\leq k,\ x_ix_{i-1}<0 \text{ and }
|x_i|=2\}.
\end{align*}
We also put
$d(\omega_1)=\sum_{i\in S_1\cup S_2} i$.
\end{notation}

\begin{proof}
We proceed by induction on the length of $\omega_1$, and on
the sum
$d=d(\omega_1)$. The property is trivially true
whether $|\omega_1|\in\{0,1,2\}$ and $d\in\N$.

Let $k$ be a positive integer such that $k\geq 2$. Suppose
that the property
holds for any word of length $k$. Suppose that
$|\omega_1|=k+1$. The property holds whether $d$ belongs to
$\{0,1\}$.
Suppose that the property holds for any word $\omega$ of
length $k+1$ such
that
$d>d(\omega)\geq 1$. Let $i$ be in $S_1\cup S_2$. If $x_i=1
$ (resp.
$x_i=-1$), then we use the rule \eqref{5} (resp. \eqref
{6}), and the length of
$\omega_1$ strictly decreases. Suppose now that $x_i=2$.
Then
$\omega_1$ can be written as
$\omega_2\, ST^x\, ST^2\, ST^y\, \omega_3$, where $x\in\Z$,
$y\in\Z_-^*$, $\omega_2$ and $\omega_3$ are words on the
alphabet
$\Sigma$, such that $\omega_2$ does not end by a $S$ and
$\omega_3$ is
either 1 or starts by $S$. If we use the rule~\eqref{3},
then we get the word
$\omega'_1=\omega_2\, ST^{x+1}\, ST^{-2}\, ST^{y+1}\,
\omega_3$.
Put $d'=d(\omega'_1)$, and
$\delta=\sum_{j\in S_1\cup S_2\setminus\{i+1,i,i-1\}} j$.
Notice that if $x=0$,
then $\omega_1=T^z\, ST^2\, ST^y\, \omega_3$, with $z\in\Z$.

\setcounter{case}{0}
\begin{case}
Suppose that either $x=-1$ or $y=-1$.
\end{case}
Then $|\omega'_1|=k$, and the rewriting process terminates.

\begin{case}
Suppose that $y\notin\{-3,-2,-1\}$.
\end{case}
\begin{itemize}
\item Suppose that $x\notin\{-2,-1,1\}$: then $d'=d-1$.
\item Suppose that $x=-2$. Then the following equalities
hold:

\begin{align*}
&\omega_1=\omega_2\, ST^{-2}\, ST^2\, ST^y\,\omega_3;\\
&\omega'_1=\omega_2\, ST^{-1}\, ST^{-2}\, ST^{y+1}\,
\omega_3.
\end{align*}

Thus, we have $d'=i+1<d=2i+1+\delta$.
\item Suppose that $x=1$. Then the following equalities
hold:

\begin{align*}
&\omega_1=\omega_2\, ST^{1}\, ST^2\, ST^y\,\omega_3;\\
&\omega'_1=\omega_2\, ST^{2}\, ST^{-2}\, ST^{y+1}\,
\omega_3.
\end{align*}

Thus we have $d'=i+1<d=2i+1+\delta$.
\end{itemize}

\begin{case}
Suppose that $y=-3$.
\end{case}
\begin{itemize}
\item Suppose that $x\notin\{-2,-1,1\}$.
Then $d'\leq i-1+\delta< i+\delta\leq d$.
\item Suppose that $x=-2$. Then

\begin{align*}
&\omega_1=\omega_2\, ST^{-2}\, ST^2\, ST^{-2}\,\omega_3;\\
&\omega'_1=\omega_2\, ST^{-1}\, ST^{-2}\, ST^{-2}\,
\omega_3.
\end{align*}

Thus $d'\leq 2i-1+\delta< 2i+1+\delta\leq d$.
\item Similarly, if $x=1$, then the following equalities
hold:

\begin{align*}
&\omega_1=\omega_2\, ST^{1}\, ST^2\, ST^{-3}\,\omega_3\\
&\omega'_1=\omega_2\, ST^{2}\, ST^{-2}\, ST^{-2}\, \omega_3.
\end{align*}

Thus $d'\leq 2i-1+\delta< 2i+1+\delta\leq d$.
\end{itemize}

\begin{case}
Suppose that $y=-2$.
\end{case}
\begin{itemize}
\item Suppose that $x\notin\{-2,-1,1\}$.
Then $d'\leq i-1+\delta< i+\delta\leq d$.
\item Suppose that $x=-2$. Then the following equalities
hold:

\begin{align*}
&\omega_1=\omega_2\, ST^{-2}\, ST^2\, ST^{-2}\,\omega_3;\\
&\omega'_1=\omega_2\, ST^{-1}\, ST^{-2}\, ST^{-1}\,
\omega_3.
\end{align*}

Thus $d'\leq 2i-1+\delta< 2i+1+\delta\leq d$.
\item Suppose that $x=1$. Then the following equalities
hold:

\begin{align*}
&\omega_1=\omega_2\, ST^{1}\, ST^2\, ST^{-2}\,\omega_3\\
&\omega'_1=\omega_2\, ST^{2}\, ST^{-2}\, ST^{-1}\, \omega_3.
\end{align*}

Thus $d'\leq 2i-1+\delta< 2i+1+\delta\leq d$.
\end{itemize}
By induction hypothesis, the rewriting process always
terminates.
\end{proof}

\begin{lemma}\label{L:redgaus}
Let $B$ be the matrix of a proper basis $(b_1,b_2)$ (see
\eqref{rows},
\eqref{b*} and \eqref{proper}). Let $x\in \Z^*$ be a non
zero integer and
$\tilde{B} := ST^x B$ express the matrix of the basis
$(\tilde{b}_1,\tilde{b}_2)$, \emph{i.e.} $\tilde{b}_1$ is
the first row of
$\tilde{B}$ and $\tilde{b}_2$ is its second row.
\begin{enumerate}
\item If $|x|\geq 3$, then $\tilde{B}$ is still proper.
Moreover, \begin{itemize}
\item if $(b_1,b_2)$ and $x$ are both positive or both
negative, then $\tilde{B}$ is proper whenever $|x|\geq 2$.
\item if $B$ is reduced, $|b_1|<|b_2|$ and $m\ne-1/2$, then
$\tilde{B}$ is
proper
for all $x\in \Z $.
\end{itemize}
\item If $|x|\geq 2$, then $|\tilde{b}_2| <|\tilde{b}_1|$.
Moreover, if $(b_1,b_2)$ and $x$ are both positive or both
negative, it is
true provided that $|x|\geq 1$.
\item If $ |x|\geq 2$, then $\max (|\tilde{b}_1|,|\tilde{b}
_2|)\geq
\max(|b_1|,|b_2|)$.
\item If $|x|\geq 1$, then $(\tilde{b}_1,\tilde{b}_2)$ and
$x$ are both
positive or both negative.
\end{enumerate}
\end{lemma}

\begin{proof}
The definitions of $B$, $b^*$ and $m$ are given by
Relations \eqref{rows},
\eqref{b*} and \eqref{proper}. We have $\tilde{b}_1= b_2
+xb_1$ and
$\tilde{b}_2=b_1$. In the sequel we express all the needed
quantities in the
orthogonal basis $b^*$.
\begin{enumerate}
\item We will show that $
\left\vert \frac{(\tilde{b}_1,\tilde{b}_1)}{(\tilde{b}
_1,\tilde{b}_2)} \right\vert \leq 2$.

\pn Notice that
$$\frac{(\tilde{b}_1,\tilde{b}_1)} {|(\tilde{b}_1,\tilde{b}
_2)|}=
\frac{(m+x)^2 |b_1^*|^2 + |b_2^*|^2}{|m+x| |b_1^*|^2} <
\frac {1}{|m+x|}.$$
Since $ ||x|-|m|| < |m+x|$ and $|m|\leq 1/2$, when $|x|\geq
3$,
clearly $2<|m+x|$ and when $mx>0$ this is still true
whenever $|x|\geq 2$.
\item We have to show that $ |b_1| < |b_2 +x b_1|$.
We notice that
$$\frac{|b_1|^2}{|b_2 +x b_1|^2} =\frac{|b_1^*|^2}{(m+x)^2
|b_1^*|^2 +
|b_2^*|^2}
 < \frac {1}{(m+x)^2}.$$
If $mx<0$, we have $1< 1.5 \leq ||x|-|m|| < |m+x|$.
If $mx>0$, $|m+x|\geq 1$, $\forall x\geq 1$.
\item From the last point, we know that $|b_1|\leq
|b_2+xb_1|$.
If $|b_2|< |b_1|$, we have also $|b_2|< |b_2+xb_1|$. (This
is true for
$|x|\geq 2$ or for $|x|\geq 1$ provided that $xm>0$.
\pn Now if $|b_2|\geq |b_1|$, since $(b_1,b_2)$ is proper,
it is indeed
Gauss--reduced. So $b_1$ and $b_2$ are the shortest vector
generating the
whole lattice. So for all $x\neq 0$, $|b_2+xb_1|>b_2$.
\item By definition $(\tilde{b}_1,\tilde{b}_2):=
(b_2+xb_1,b_1)=
(x+m)|b_1^*|^2$. Since $|m|\leq 1/2$ if $|x|\geq 1$, the
quantities $x+m$
and $x$
have the same sign.
\end{enumerate}
\end{proof}

\begin{corollary}\label{C:redgaus}
Let $\omega_1$ be a word of the form~\eqref{bondepart}.
Then there is a
unique reduced word $\omega_1'$ corresponding to $\omega_1$.
Moreover, $\omega'_1$ is a Gaussian word.
\end{corollary}

\begin{lemma}\label{L:redST2STx}
Let $B$ be a proper basis. Put
 \begin{align*}
 &\tilde{B}:=ST^2\,ST^x\,B
 \text{ and }
 B':=(T\, ST^{-2}\, ST^{x+1})^{-1}\tilde{B};\\
 &(\text{resp. }\tilde{B}:=ST^{-2}\,ST^x\,B
 \text{ and }
 B':=(T^{-1}\, ST^{2}\, ST^{x-1})^{-1}\tilde{B}).
 \end{align*}
Then $B'$ is such that $B'=-B$.
\end{lemma}

\begin{proof}
The following equalities hold:
 \[
 ST^2\,ST^x=\begin{pmatrix} 2x+1 & 2 \\ x & 1 \end{pmatrix}
 \text{ and }
 T\,ST^{-2}\,ST^{x+1}=\begin{pmatrix} -2x-1 & -2 \\ -x & -1
\end{pmatrix},
 \]
which concludes the proof.
\end{proof}

\begin{lemma}\label{L:redSTSTx}
Let $R$ be a reduced basis such that $m\ne-\frac{1}{2}$. Put
$\tilde{B}:=ST\,ST^x\,\omega\,B$
(resp. $\tilde{B}:=ST^{-1}\,ST^x\,\omega_1\,B$),
where $\omega$ is a word on $\Sigma$ of the form~\eqref
{bondepart}. Then
the basis
$B':=(T\, ST^{-x-1}\,\omega^-)^{-1}\tilde{B}$ (resp.
$B':=(T^{-1}\, ST^{-x+1}\,\omega^-)^{-1}\tilde{B}$) is a
reduced basis.
\end{lemma}

\begin{proof}
\pn Notice that:
 \[
 ST\,ST^x=\begin{pmatrix} x+1 & 1 \\ x & 1 \end{pmatrix}
 \text{ and }
 T\,ST^{-x-1}=\begin{pmatrix} -x-1 & 1 \\ -x & 1 \end
{pmatrix}.
 \]
Thus, it is easy to prove that
$ST\,ST^x=T\,ST^{-x-1}\,\begin{pmatrix} -1 & 0 \\ 0 & 1 \end
{pmatrix}$.

The proof of the following easy claim is left to the reader.
\begin{claim}
Let $y$ be an integer. Then the following equalities hold:
$$
\begin{pmatrix} -1 & 0 \\ 0 & 1 \end{pmatrix}\,ST^y
=ST^{-y}\,\begin{pmatrix} 1 & 0 \\ 0 & -1 \end{pmatrix}
\text{ and }
\begin{pmatrix} 1 & 0 \\ 0 & -1 \end{pmatrix}\,ST^y
= ST^{-y}\,\begin{pmatrix} -1 & 0 \\ 0 & 1 \end{pmatrix}.
$$
\end{claim}

Thus, it is easy to prove by induction on the length of
$\omega_1$ that
there exists an integer $\beta\in\{-1,1\}$ such that
$ST\,ST^x\omega_1=
T\,ST^{-x-1}\,\omega^-_1\,
\beta\begin{pmatrix} -1 & 0 \\ 0 & 1 \end{pmatrix}$.
Thus, we have $B'=\beta\begin{pmatrix} -1 & 0 \\ 0 & 1 \end
{pmatrix}\, B$,
which
concludes the proof.
\end{proof}

\begin{corollary}\label{C:GLconcl}
Let $\omega_1=T^{x_{k+1}}\,\prod_{i=k}^1 ST^{x_{i}}$ be a
non--reduced
word, and let $R$ be a reduced basis. If
$\omega_1\, B$ is the input of the Gaussian algorithm, then
the sequence of
elementary transforms made during the execution is exactly
represented by the
reduced word uniquely associated to $\omega_1$.
\end{corollary}

\section{The length of a unimodular matrix with respect to
\\its Gaussian decomposition}

The easy proof of the following lemma is left to the reader.
\begin{lemma}\label{L:sizemin1}
Let $k$ be a positive integer. There exist nonnegative
integers $\alpha$,
$\beta$ and $\gamma$ such that:
 \[
 (ST^2)^k=\begin{pmatrix}\alpha&\beta\\\beta&\gamma\end
{pmatrix},
 (ST^{-2})^k=(-1)^k\begin{pmatrix}\alpha&-\beta\\-
\beta&\gamma\end{pmatrix}
\text{ and }\alpha=2\beta+\gamma.
 \]
\end{lemma}

\begin{lemma}\label{L:sizemin2}
Let $B$ be a proper basis, and let $x$ be an integer such
that $|x|\geq3$. Then
$\length(ST^x\,B)\geq\length(ST^2\,B)$ and
$\length(ST^x\,B)\geq\length(ST^{-2}\,B)$.
\end{lemma}

\begin{proof}
Put $\tilde{B}_x=(\tilde{b}_{1,x},\tilde{b}_{2,x})
=ST^x\,B$, for all
integer $x$.
Then
\begin{itemize}
\item $\tilde{b}_{1,x}=(x+m)b_1^*+b_2^*$;
\item $\tilde{b}_{2,x}=b_1^*$.
\end{itemize}

It is obvious that
$|\tilde{b}_{2,x}|^2=|\tilde{b}_{2,2}|^2=|\tilde{b}_{2,-2}
|^2$. Moreover,
we have
\[
|\tilde{b}_{1,x}|^2-|\tilde{b}_{1,2}|^2
=[(x+m)^2-(2+m)^2]|b_1^*|^2=[(x+2+2m)(x-2)]|b_1^*|^2\geq 0,
\]
and
\[
|\tilde{b}_{1,x}|^2-|\tilde{b}_{1,-2}|^2=[(x+m)^2-(m-2)^2]
|b_1^*|^2
=[(x-2+2m)(x+2)]|b_1^*|^2\geq 0,
\]
which concludes the proof.
\end{proof}

\begin{lemma}\label{L:sizemin3}
Let $B$ be a proper basis, let $k$ be a positive integer,
let
$\varepsilon\in\{1,-1\}$ be an integer, and let
$x$ be an integer such that $|x|\geq3$. The following
properties hold:
\begin{itemize}
\item if $x$ is positive, then
$\length((ST^2)^k\,ST^x\,B)\geq\length((ST^{\varepsilon2})^
{k+1}\,B)$;
\item if $x$ is negative, then
$\length((ST^{-2})^k\,ST^x\,B)\geq\length((ST^
{\varepsilon2})^{k+1}\,B)$.
\end{itemize}
\end{lemma}

\begin{proof}
Put $(ST^2)^k=\begin{pmatrix}\alpha&\beta\\\beta&\gamma\end
{pmatrix}$ and
$(ST^{-2})^k=(-1)^k\begin{pmatrix}\alpha&-\beta\\-
\beta&\gamma\end{pmatrix}$, as
in
Lemma~\ref{L:sizemin1}. Put
$\tilde{B}=(\tilde{b}_{1,x},\tilde{b}_{2,x})=(ST^2)
^k\,ST^x\,B$ and
$\tilde{B'}=(\tilde{b'}_{1,x},\tilde{b'}_{2,x})$.
Then
\begin{itemize}
\item $\tilde{b}_{1,x}=(\alpha(x+m)+\beta)b_1^*+\alpha
b_2^*$,
 \qquad $\tilde{b'}_{1,x}=(\alpha(x+m)-\beta)
b_1^*+\alpha b_2^*$;
\item $\tilde{b}_{2,x}=(\beta(x+m)+\gamma)b_1^*+\beta
b_2^*$,
 \qquad $\tilde{b'}_{2,x}=(-\beta(x+m)+\gamma)b_1^*-
\beta b_2^*$;
\end{itemize}
Suppose that $x$ is positive. Then
\begin{multline*}
|\tilde{b}_{1,x}|^2-|\tilde{b}_{1,2}|^2=
[(\alpha(x+m)+\beta)^2-(\alpha(2+m)+\beta)^2]|b_1^*|^2\\
=[\alpha(x-2)(\alpha(x+2+2m)+2\beta)]|b_1^*|^2\geq 0
\end{multline*}
and
\begin{multline*}
|\tilde{b}_{2,x}|^2-|\tilde{b}_{2,2}|^2=
[(\beta(x+m)+\gamma)^2-(\beta(2+m)+\gamma)^2]|b_1^*|^2\\
=[\beta(x-2)(\beta(x+2+2m)+2\gamma)]|b_1^*|^2\geq 0,
\end{multline*}
that is, $\length((ST^2)^k\,ST^x\,B)\geq\length((ST^2)^{k+1}
\,B)$.
Moreover,
\begin{multline*}
|\tilde{b}_{1,x}|^2-|\tilde{b'}_{1,-2}|^2=
[(\alpha(x+m)+\beta)^2-(\alpha(m-2)-\beta)^2]|b_1^*|^2\\
=[\alpha(x-2+2m)(\alpha(x+2)+2\beta)]|b_1^*|^2\geq 0
\end{multline*}
and
\begin{multline*}
|\tilde{b}_{2,x}|^2-|\tilde{b'}_{2,-2}|^2=
[(\beta(x+m)+\gamma)^2-(-\beta(m-2)+\gamma)^2]|b_1^*|^2\\
=[\beta(x-2+2m)(\beta(x+2)+2\gamma)]|b_1^*|^2\geq 0,
\end{multline*}
that is, $\length((ST^2)^k\,ST^x\,B)\geq\length((ST^{-2})^
{k+1}\,B)$. The
second part
of the proof is very similar.
\end{proof}

Proofs of Lemmas~\ref{L:sizemin4} and~\ref{L:sizemin5} are
very similar to
the proof of
Lemma~\ref{L:sizemin3}.

\begin{lemma}\label{L:sizemin4}
Let $B$ be a proper basis and let $x$ be an integer. Then
$\length(T^x\,B)\geq\length(B)$.
\end{lemma}

\begin{lemma}\label{L:sizemin5}
Let $R$ be a reduced basis, let $k$ be a positive integer,
let
$x$ and $x'$ be integers. Then
\begin{itemize}
\item if $x\geq x'\geq 0$, then
$\length((ST^{2})^k\,ST^x\,R)\geq\length((ST^{2})^k\,ST^{x'}
\,R)$;
\item if $x\leq x'\leq 0$, then
$\length((ST^{-2})^k\,ST^x\,R)\geq\length((ST^{-2})^k\,ST^
{x'}\,R)$.
\end{itemize}
Moreover, the following properties hold:
\begin{itemize}
\item if $m$ is positive, then
$\length((ST^2)^k\,S\,R)\geq\length((ST^{-2})^{k}\,S\,R)$;
\item if $m$ is negative, then
$\length((ST^{-2})^k\,S\,R)\geq\length((ST^{2})^{k}\,S\,R)$.
\end{itemize}
\end{lemma}

\begin{proof}
Proof of Theorem~\ref{T:basemin}.

Notice first that, by Lemma~\ref{L:sizemin4}, we can
suppose that $x_{k+1}=0$.
Suppose that $k=1$. Then, by Lemma~\ref{L:sizemin5},
$\length(ST^{x_1}B)\geq\length(SB)$.

Suppose that $k\geq 1$. The variable $\varepsilon$ will
denote an integer in
$\{1,-1\}$. We construct a sequence $\omega_0$, $\omega_1$,
\dots, $\omega_k$
such that the following properties hold:
\begin{enumerate}
\item $\omega_0=\omega$ and $\omega_k=(ST^{\varepsilon2})^
{k-1}S$;
\item $\omega_j=(ST^{\varepsilon 2})^j\prod_{i=1}^{k-j}ST^
{x_i}$, with
$\varepsilon\in\{1,-1\}$ such that $(\varepsilon 2)x_{k-j}
\geq 0$, for
all $j\in\{1,\dots,k-1\}$;\label{E:motST}
\item $\length(\omega_j\,B)\leq\length(\omega_{j-1}\,B)$,
for all
$j\in\{1,\dots,k\}$.\label{E:motlong}
\end{enumerate}
Notice that $\omega_j$ may be equal to $\omega_{j-1}$ for
some
$j\in\{1,\dots,k\}$.
We proceed by induction on an integer $j$ such that $1\leq
j\leq k$.
Put $\omega'=\prod_{i=1}^{k-1}ST^{x_i}$, $\tilde{B'}
=\omega'B$ and
$\tilde{B}=ST^{x_k}\tilde{B'}=\omega\,B$. By Lemma~\ref
{L:sizemin2}, we
know that
$\length(\tilde{B})=\length(ST^{x_k}\tilde{B'})
\geq\length(ST^{\varepsilon2})\tilde{B'})$. Now, either $x_
{k-1}$ is
positive and we
put $\omega_1=ST^2\omega'$, or $x_{k-1}$ is negative and we
put $\omega_1=ST^{-2}\omega'$, so that $\omega_1$ is
Gaussian.

Let $j$ be an integer in $\{1,\dots,k-2\}$ such that
$\omega_j$
verifies~\eqref{E:motST} and~\eqref{E:motlong}. By
Lemma~\ref{L:sizemin3},
it is
obvious that we can put:
\begin{itemize}
\item $\omega_{j+1}=(ST^2)^{j+1}\prod_{i=1}^{k-j-1}ST^{x_i}
$ if $x$ is
positive;
\item $\omega_{j+1}=(ST^{-2})^{j+1}\prod_{i=1}^{k-j-1}ST^
{x_i}$ if $x$ is
negative.
\end{itemize}

Suppose now that we constructed
$\omega_{k-1}=(ST^{\varepsilon2})^{k-1}ST^{x_1}$.
Then, by Lemma~\ref{L:sizemin5}, we can put $\omega_k=(ST^
{2})^{k-1}S$
whether $m$ is
positive and $\omega_k=(ST^{-2})^{k-1}S$ whether $m$ is
negative.
\end{proof}
\section{The maximum number of steps of the Gaussian
algorithm}

\begin{proof}
Proof of Theorem~\ref{T:compirecas}.

 From the last section , we know that the input with the
smallest length
demanding $k$ steps is $(ST^2)^k$ or $(ST^{-2})^k$. Let $Q$
be equal to
$(ST^2)$.
The symmetrical matrix $Q$
$$\begin{pmatrix}2 & 1\\1 &0\end{pmatrix}$$
has its eigenvalues equal to $1+\sqrt 2$ and $1-\sqrt 2$.
It can be
diagonalized and one deduces that there exist two constants
$\alpha>0$ and
$\beta>0$ such that all coefficients of $Q^k$ are expressed
in the form
$$ \alpha (1+\sqrt 2)^k +\beta (1-\sqrt 2)^k =
\alpha (1+\sqrt 2)^k
\left(1 + \frac{\beta (1-\sqrt 2)^k}{\alpha (1+\sqrt 2)^k}
\right).$$ This leads
to the lower--bound proposed by the theorem.
\end{proof}

\begin{proof}
Proof of Corollary~\ref{C:conspircas}.
\end{proof}

\section{Sorting algorithms}

\begin{lemma}\label{L:bubblerpt}
The rewriting process always terminates.
\end{lemma}

\begin{proof}
Let $l\colon\Sigma^*\to\N$ be the map given by the rule
 \[
 l(\omega)=\sum_{i=1}^{|\omega|} \alpha_i,
 \]
where $\omega=t_{\alpha_1}\dots t_{\alpha_{|\omega|}}$.
Let $h\colon\Sigma^*\to\N$ be the map given by the rule
 \[
 h(\omega)=\sum_{i=1}^{s(\omega)} (s(\omega)-i)(\max
(\omega)-|\omega_i|),
 \]
where $[\omega_1,\dots,\omega_m]$ is the increasing
decomposition of $\omega$.
We proceed here by induction on the nonnegative integers $l
(\omega)$ and $h(\omega)$. Suppose that $l(\omega)=0$. Then
$\sigma=1$ and the rewriting process is over. Suppose that
$l(\omega)=1$ and $h(\omega)=0$. Then $\omega=t_1$ and the
rewriting process is over. Notice that there exists no
word on $\Sigma^*$ such that $l(\omega)=1$ and $h(\omega)>0
$.

Suppose now that $l(\omega)=l$, that $h(\omega)=h$ and
that we can apply a rewriting rule. If we apply one of the
rules \eqref{perm1} or \eqref{perm2}, then we get a word
$\omega'$ such that $l(\omega')<l(\omega)$. If we apply
\eqref{perm3} or \eqref{perm4}, then $l(\omega')=l(\omega)
$. Let $[\omega_1,\dots,\omega_m]$ be the increasing
decomposition of $\omega$.

Suppose that we can apply \eqref{perm3}. Then there exists
$j\in\{1,\dots,m-1\}$ such that $\min(\omega_{j+1})=i$ and
$\min(\omega_{j})>i+1$. Put
 \begin{align*}
 \omega=\omega_1\,\dots\,\omega_j\, t_i\,\omega'_{j+1}
\,\dots\,\omega_m;\\
 \omega'=\omega_1\,\dots\, t_i\,\omega_j\, \omega'_{j+1}
\,\dots\,\omega_m;
 \end{align*}
where $\omega'_{j+1}$ is such that
$\omega_{j+1}=t_i\,\omega'_{j+1}$. Then
 \[
 \begin{tabular}{l l l l}
 $h(\omega')$ & = & $h(\omega)$& $-(s(\omega)-j)(\max
(\omega)-|\omega_j|)$
 $-(s(\omega)-(j+1))(\max(\omega)-|\omega_{j+1}|)$\\
 & & & $+(s(\omega)-j)(\max(\omega)-(|\omega_j|+1))$\\
 & & & $+(s(\omega)-(j+1))(\max(\omega)-(|\omega_{j+1}|-
1))$\\
 & = & $h(\omega)$ & $-(s(\omega)-j)+(s(\omega)-(j+1))$\\
 & = & $h(\omega)$ & $-1$.\\
 \end{tabular}
 \]
Suppose now that we can apply the rule~\eqref{perm4}. Then
there exist $p$, $q\in\{1,\dots,m\}$ and $i$, $j$,
$k\in\{1,\dots,\max(\omega)\}$ such that one of the
following cases occurs.
\setcounter{case}{0}
\begin{case}
Suppose that $i<k\leq\max(\omega_p)$.
\end{case}
Then
 \begin{align*}
 \omega=\omega_1\,\dots\,\omega'_p\, t_i\,t_k\,\omega''_
{p}\,\dots\,
 \omega'_q\, t_j\,\omega''_{q}\,\dots\,\omega_m;\\
 \omega'=\omega_1\,\dots\,\omega'_p\,
t_i\,t_j\,t_k\,\omega''_{p}\,\dots\,
 \omega'_q\,\omega''_{q}\,\dots\,\omega_m.
 \end{align*}
\begin{case}
Suppose that $i=\max(\omega_p)$ and $k=\min(\omega_{p+1})$.
\end{case}
Then,
 \begin{align*}
 \omega=\omega_1\,\dots\,\omega'_p\, t_i\,t_k\,\omega'_
{p+1}\,\dots\,
 \omega'_q\, t_j\,\omega''_{q}\,\dots\,\omega_m;\\
 \omega'=\omega_1\,\dots\,\omega'_p\,
t_i\,t_j\,t_k\,\omega'_{p+1}\,\dots\,
 \omega'_q\,\omega''_{q}\,\dots\,\omega_m.
 \end{align*}

Thus,
 \[
 \begin{tabular}{l l l l}
 $h(\omega')$ & = & $h(\omega)$& $-(s(\omega)-p)(\max
(\omega)-|\omega_p|)$
 $-(s(\omega)-q)(\max(\omega)-|\omega_{q}|)$\\
 & & & $+(s(\omega)-p)(\max(\omega)-(|\omega_p|+1))$\\
 & & & $+(s(\omega)-q)(\max(\omega)-(|\omega_{q}|-1))$\\
 & = & $h(\omega)$ & $-(s(\omega)-p)+(s(\omega)-q)$\\
 & = & $h(\omega)$ & $+p-q$,\\
 \end{tabular}
 \]
which concludes the proof.
\end{proof}

\begin{lemma}\label{L:iwp}
Let $\omega$ be a reduced word, and let
$[\omega_1,\dots,\omega_m]$ be its increasing
decomposition. Let $p\in\{1,\dots,m-1\}$ be an integer. If
$t_i$ is in $\omega_{p+1}$, then $t_{i+1}$ is in $\omega_p$.
\end{lemma}

\begin{proof}
Let $p\in\{1,\dots,m\}$ be such that the property does
not hold. Let $i\in\{1,\dots,\max(\omega)\}$ be the
smallest integer such that $t_i\in\omega_{p+1}$, and $t_
{i+1}\notin\omega_p$. Then, there exists $\omega'_p$,
$\omega''_p$, $\omega'_{p+1}$, $\omega''_{p+1}\in\Sigma^*$
such that
 \[
 \omega=\omega_1\,\dots\,\omega'_p\,\omega''_{p}\,
 \omega'_{p+1}\,t_i\,\omega''_{p+1}\,\dots\,\omega_m;
 \]
and such that the following inequalities hold:
 \begin{align*}
 &\max(\omega'_{p})<i+1<\min(\omega''_p);\\
 &\max(\omega'_{p+1})<i<\min(\omega''_{p+1});\\
 &\dist(t_i,\omega''_{p})>1.
 \end{align*}

\setcounter{case}{0}
 \begin{case}
 Suppose that $i=\min(\omega_{p+1})$.
 \end{case}
Then $\omega'_{p+1}=1$, $\omega''_{p}$ is maximal on the
left and we can apply the rule~\eqref{perm3}.
 \begin{case}
 Suppose that $i\ne\min(\omega_{p+1})$ and that $t_{i-1}
\in\omega''_{p+1}$.
 \end{case}
Since $t_i\in\omega_p$ and $\dist(t_i,\omega''_p)>1$, we
can apply the rule~\eqref{perm2}.
 \begin{case}
 Suppose that $i\ne\min(\omega_{p+1})$ and that $t_{i-1}
\notin\omega''_{p+1}$.
 \end{case}
Then $\dist(t_i,\omega''_p\omega'_{p+1})>1$. Thus, either
$\omega'_p=1$ and we can apply the rule~\eqref{perm3}, or
$\omega'_p\ne 1$ and we can apply the rule~\eqref{perm4}.
\end{proof}

The following corollary is an easy consequence of Lemma~\ref
{L:iwp}.

\begin{corollary}\label{C:algow1}
Let $\omega$ be a reduced word, and let
$[\omega_1,\dots,\omega_m]$ be its increasing
decomposition. Let $p\in\{1,\dots,m-1\}$ be an integer. Then
 \begin{enumerate}
 \item $\max(\omega_p)>\max(\omega_{p+1})$;
 \item $\dist(\omega_p,\omega_{p+1})\leq 1$.
 \end{enumerate}
\end{corollary}
 
The easy but tedious proof of the following lemma is left
to the reader, who can proceed by induction.

\begin{lemma}\label{L:wordmax}
Let $\omega=\omega_1\dots\omega_m$ be a permutation put
under increasing decomposition. Put
$\omega[1,\dots,n]=[\alpha_1,\dots,\alpha_n]$ and
$\omega_1=\prod_{k=1}^q t_{i_k}\dots t_{i_k+p_k}$\footnote{
$\prod_{k=1}^q t_{i_k}\dots t_{i_k+p_k}$ is the
consecutively increasing decomposition of $\omega_1$.}.
Suppose that $\omega$ is reduced.
Then the following properties hold:
\begin{enumerate}
\item $\alpha_1<\alpha_2<\dots<\alpha_{i_1}$;
\item $\alpha_{i_k}=\max\{\alpha_1,\alpha_2,\dots,\alpha_
{i_k+p_k}\}$;
\item $\alpha_{i_k}<\alpha_{i_k+p_k+1}<\alpha_{i_k+p_k+2}
<\dots<\alpha_{i_{k+1}}$;
\end{enumerate}
for all $k\in\{1,\dots,q\}$.
\end{lemma}

\begin{corollary}
In the conditions of Lemma~\ref{L:wordmax}, the first
iterations
of the algorithm on $\omega[1,\dots,n]$ are represented by
$\omega_1$.
\end{corollary}

The following corollary is an easy consequence of Lemma~\ref
{L:wordmax}.

\begin{corollary}\label{C:bub=red}
Let $\omega$ be a reduced word. Then $\omega$ is a bubblian
word.
\end{corollary}

Since the bubblian word associated to a given permutation
is unique and the rewriting process terminates,
the bubblian words are the reduced words.

\subsection{Insertion sort}
The basic idea of the insertion sort algorithm is the
following: the beginning of the list being already sorted,
the first non sorted element of the list is put at the
right place in the already sorted part. Thus, an elementary
transform made by the algorithm is a cycle
$(i,i+1,\dots,i+p)=t_{i+p}\,t_
{i+p-1}\,\dots\,t_{i+1}\,t_i$ in $\mathcal{S}$. Thus, any
permutation $\sigma$ can be written on the form
 \[
 \sigma=\prod_{p=1}^{m}(i_p,\dots,i_p+q_p)=\prod_{p=1}^
{m}t_{i_p+q_p}\,\dots\,t_{i_p},
 \]
and the word on $\Sigma^*$ produced by the algorithm, which
we will call \emph{insertion word}, is a particular
decomposition of $\sigma$. Let $\omega$ be a word on
$\sigma^*$. The rewriting rules for the insertion sort
algorithm are very similar to the rewriting rules for the
bubble sort algorithm. In the following equations, $i$, $j$
and $k$ are arbitrary positive integers and $\omega$ is a
word on $\Sigma$:
 \begin{equation}\label{permins1}
 \text{ if } \dist(i,\omega)>1, \text{ then }
 t_i\,\omega\,t_i \longrightarrow \omega;
 \end{equation}
 \begin{equation}\label{permins2}
 \text{ if } \dist(i+1,\omega)>1, \text{ then }
 t_{i+1}\,t_{i}\,\omega\,t_{i+1}\longrightarrow t_{i}
\,t_{i+1}\,t_{i}\,\omega;
 \end{equation}
 \begin{equation}\label{permins3}
 \text{ if } \dist(i+1,\omega)>1 \text{, then }
 t_{i+1}\,\omega\,t_{i}\longrightarrow \omega\,t_{i+1}
\,t_i;
 \end{equation}
 \begin{equation}\label{permins4}
 \text{ if } j-i>1,\text{ then }
 t_{j+1}\,t_j\,t_{i}\longrightarrow t_{j+1}\,t_i\,t_j.
 \end{equation}
Let $\omega$ be a word on $\sigma^*$.
Using similar methods that in Subsection~\ref{Ss:bubble},
we can prove that, by applying repeatedly rules~\ref
{permins1} to~\ref{permins4}, we obtain the insertion
word uniquely associated to $\omega$.

\subsection{Selection sort}
The basic idea of the selection sort algorithm is the
following: the beginning of the list being already sorted,
the algorithm finds the smallest element of the non sorted
list and put it at the right place at the end of the
already sorted part. As in the previous subsection, an
elementary transform made by the algorithm is a cycle
$(i,i+1,\dots,i+p)$, which can be written as $t_{i+p}\,t_
{i+p-1}\,\dots\,t_{i+1}\,t_i$ in $\Sigma^*$. The word on
$\Sigma^*$ produced by the algorithm, called \emph
{selection word}, is a particular decomposition of
$\sigma$.

Let $\omega$ be a word on $\sigma^*$. The rewriting rules
for the selection sort algorithm are very similar to the
rewriting rules for the bubble sort algorithm. In the
following equations, $i$, $j$ and $k$ are arbitrary
positive integers and $\omega$ is a word on $\Sigma$:
 \begin{equation}\label{permsel1}
 \text{ if } \dist(i,\omega)>1, \text{ then }
 t_i\,\omega\,t_i \longrightarrow \omega;
 \end{equation}
 \begin{equation}\label{permsel2}
 \text{ if } \dist(i,\omega)>1, \text{ then }
 t_{i}\,\omega\,t_{i+1}\,t_{i}\longrightarrow
\omega\,t_{i+1}\,t_{i}\,t_{i+1};
 \end{equation}
 \begin{equation}\label{permsel3}
 \text{ if } \dist(i,\omega)>1 \text{, then }
 t_{i+1}\,\omega\,t_{i}\longrightarrow t_{i+1}
\,t_i\,\omega;
 \end{equation}
 \begin{equation}\label{permsel4}
 \text{ if } i-j>1,\text{ then }
 t_{i}\,t_{j}\,t_{j-1}\longrightarrow t_{j}\,t_i\,t_
{j-1}.
 \end{equation}
Let $\omega$ be a word on $\sigma^*$.
Using similar methods that in Subsection~\ref{Ss:bubble},
we can prove that, by applying
repeatedly rules~\ref{permins1} to~\ref{permins4}, we can
obtain the insertion word uniquely associated to $\omega$.
\medskip


\begin{thebibliography}{VGT88}

\bibitem{Akh00}
{\sc A.~Akhavi.}
\newblock Worst--case complexity of the optimal {LLL}
algorithm.
\newblock In {\em Proceedings of LATIN'2000 - Punta del
Este}.
\newblock LNCS 1776, pp 476--490.

\bibitem{DaFlVa}
{\sc H.~Daud\'e, Ph. Flajolet, and B.~Vall\'ee.}
\newblock An average-case analysis of the {G}aussian
algorithm for lattice
 reduction.
\newblock {\em Comb., Prob. \& Comp.}, 123:397--433, 1997.

\bibitem{DaVa}
{\sc Daud\'e, H., and Vall\'ee, B.}
\newblock An upper bound on the average number of
iterations of the {LLL}
 algorithm.
\newblock {\em Theoretical Computer Science 123(1)\/}
(1994), pp.~95--115.

\bibitem{JoSt98}
{\sc A.~Joux and J.~Stern.}
\newblock Lattice reduction: {A} toolbox for the
cryptanalyst.
\newblock {\em J. of Cryptology}, 11:161--185, 1998.


\bibitem{KaSc96}
{\sc M.~Kaib and C.~P. Schnorr.}
\newblock The generalized {G}auss reduction algorithm.
\newblock {\em J. of Algorithms}, 21:565--578, 1996.

\bibitem{Kann.G-Z}
{\sc R.~Kannan.}
\newblock Improved algorithm for integer programming and
related lattice
 problems.
\newblock In {\em 15th Ann. ACM Symp. on Theory of
Computing}, pages 193--206,
 1983.

\bibitem{Lag1}
{\sc J.~C. Lagarias.}
\newblock Worst-case complexity bounds for algorithms in
the theory
of integral
 quadratic forms.
\newblock {\em J. Algorithms}, 1:142--186, 1980.

\bibitem{Lag0}
{\sc J.~C. Lagarias.}
\newblock The computational complexity of simultaneous {D}
iophantine
approximation problems
\newblock {\em SIAM J. Computing}, 14:196--209, 1985.


\bibitem{Len83}
{\sc H.W. Lenstra.}
\newblock Integer programming with a fixed number of
variables.
\newblock {\em Math. Oper. Res.}, 8:538--548, 1983.




\bibitem{Len00}
{\sc H.W. Lenstra.}
\newblock Flags and lattice basis reduction.
\newblock In {\em Procedings of the 3rd European Congress of Mathematics - Barcelona July 2000}
\newblock {I: 37-51, Birkh\"auser Verlag, Basel}



\bibitem{LLL82}
{\sc A.~K. Lenstra, H.~W. Lenstra, and L.~Lov{\'a}sz.}
\newblock Factoring polynomials with rational coefficients.
\newblock {\em Math. Ann.}, 261:513--534, 1982.

\bibitem{Mag34}
{\sc W. Magnus.}
\newblock {\"U}ber $n$-dimensionale {G}
ittertransformationen.
\newblock {\em Acta Math.}, 64:353--357, 1934.

\bibitem{MKS76}
{\sc W. Magnus, A. Karrass, and D. Solitar.}
\newblock Combinatorial group theory.
\newblock Dover, New York, 1976 (second revised edition).

\bibitem{Nie24}
{\sc J. Nielsen.}
\newblock Die {G}ruppe der dreidimensionale {G}
ittertransformationen.
\newblock {\em Kgl Danske Videnskabernes Selskab., Math.
Fys. Meddelelser},
V 12: 1--29, 1924.

\bibitem{Nie24b}
{\sc J. Nielsen.}
\newblock Die {I}somorphismengruppe der freien {G}ruppen.
\newblock {\em Math. Ann.}, 91:169--209, 1924.
\newblock translated in english by J. Stillwell in {\em J.
Nielsen collected
papers}, Vol 1.

\bibitem{ValGaussRevisit}
{\sc B.~Vall\'ee.}
\newblock Gauss' algorithm revisited.
\newblock {\em J. of Algorithms}, 12:556--572, 1991.

\bibitem{vSpr}
{\sc O.~von Sprang.}
\newblock {\em Basisreduktionsalgorithmen f{\"u}r Gitter
kleiner Dimension}.
\newblock PhD thesis, Universit{\"a}t des Saarlandes, 1994.

\end{thebibliography}
\end{document}